\documentclass[a4paper]{jpconf}
\usepackage{graphicx}
\usepackage{palatino,mathrsfs}
\usepackage[mathcal]{euler}
\usepackage{mathtools}
\usepackage{epstopdf,epsfig}
\usepackage{comment} 
\usepackage{amsmath,amsthm,amsfonts,amssymb,latexsym,amscd,enumerate,url,hyperref}
\usepackage{amssymb}
\usepackage{url}
\usepackage{graphicx}
\usepackage[all,knot]{xy}
\usepackage{xcolor}
\newtheorem{theorem}{Theorem}
\xyoption{arc}

\newenvironment{dedication}
  { \itshape             
 \raggedleft
  }
  {\par 
   
  }

\newcommand{\R}{\mathbb{R}}
\newcommand{\C}{\mathbb{C}}

\newcommand{\g}{\mathfrak{g}}
\newtheorem{lemma}{Lemma}
\newtheorem{proposition}{Proposition}

\begin{document}
\title{Strong contraction, the mirabolic group and the Kirillov conjecture}

\author{E M Subag$^1$ and E M Baruch$^2$}
\address{$^1$ Department of Mathematics, Penn State University, University Park, PA 16802, USA.\\
$^2$ Department of Mathematics, Technion-Israel Institute of Technology,
Haifa 32000, Israel.}
\ead{$^1$eus25@psu.edu, $^2$embaruch@math.technion.ac.il}

\begin{abstract}
We lift any (infinitesimal)  unitary irreducible representation of $GL_n(\R)$ to a family of representations that strongly contracts to a certain type of  (infinitesimal)  unitary irreducible representations of $\mathbb{R}^n\rtimes {M}_n$, with $M_n$ being the mirabolic subgroup of $GL_n(\R)$. For the case of $n=2$ we obtain the full unitary dual of   $\mathbb{R}^2\rtimes {M}_2$ as a strong contraction. 
We demonstrate the role of the Kirillov conjecture and Kirillov model for these contractions. 
\end{abstract}

\begin{dedication}

Dedicated to I.E. Segal (1918-1998)  in commemoration of the centenary of his birth.
\end{dedication}

\section{Introduction}
Contraction of Lie algebras first appeared in the work of Segal \cite{Segal51} and  {I}n\"{o}n\"{u} and Wigner \cite{Inonu-Wigner53}. Over the years, many applications to mathematical physics where found e.g., see \cite{Saletan61,PhysRevA.6.2211,MR1275599,MR720801,BI,MR3687831} and references therein. Several years ago, the authors, together with  their collaborators, introduced the notion of \textit{strong contraction} for representations of Lie algebras \cite{Subag12}. This involves a new setup for contraction  of representations which utilize  representations that are realized on certain  spaces of functions.  Many contractions, and most notably {I}n\"{o}n\"{u}-Wigner contractions, are being preformed with respect to a subgroup (or sub-Lie-algebra). Here, we shall focus on  the  contraction of   $\mathfrak{gl}_n(\mathbb{R})$ with respect to the Lie algebra $\mathfrak{m}_n$ consisting of all   $n\times n$ real matrices having their last row  equal to zero. 
This is the Lie algebra of the mirabolic  group $M_n$, consisting of all invertible $n\times n$ matrices having their last row given by   $(0,0,...,0,1)$. 
The Lie algebra  obtained via this contraction is the semidirect product  $\mathbb{R}^{n}\rtimes \mathfrak{m}_n$. It is the Lie algebra of  $\mathbb{R}^{n}\rtimes {M}_n$. The mirabolic group $M_n$  has the following remarkable property: any  unitary irreducible representation of $GL_n(\R)$ restricted to $M_n$  remains irreducible. This is known as the \textit{Kirillov conjecture}, proven in the $\mathfrak{p}$-adic case in \cite{MR748505} and for real groups in \cite{MR1999922}. An important feature of the irreducible representations of $GL_n(\R)$ is the existence of their Kirillov model which is  a   realization on a certain space  of functions on $GL_{n-1}(\R)$ \cite{MR0139691}. The Kirillov model  arises in various contexts in automorphic representation theory  and representation theory of finite groups of Lie type, e.g., see\cite{MR2154720,MR1063847}. Recently, an explicit formulas for the space of $K$-finite vector in the  Kirillov model of $GL_2(\mathbb{R})$ was found \cite{MR3055523}. This  was used  in \cite{Subag12} to obtain the skew-Hermitian irreducible  representations of the Poincar\'e Lie algebra $\mathfrak{iso}(1,1)$ as  strong contractions of representations of $\mathfrak{sl}_2(\mathbb{R})$.

The first purpose of this paper is  to demonstrate how the Kirillov conjecture  and the Kirillov   model are useful for strong contractions in the context of  contraction of $\mathfrak{gl}_n(\R)$ with respect to $\mathfrak{m}_n$.  In particular we shall use the Kirillov conjecture to prove the following result.\\

\noindent \textbf{Theorem 1.}
\textit{Let  $\pi:GL_n(\R)\longrightarrow \mathcal{U}(\mathcal{H})$ be a unitary irreducible representation, realized on a Hilbert space of functions $\mathcal{H}$. Let $\pi_0:\mathbb{R}^{n}\rtimes {M}_n\longrightarrow \mathcal{U}(\mathcal{H})$  be the unitary irreducible representation given by the restriction of $\pi$ to $M_n$ extended  trivially to  $\mathbb{R}^{n}\rtimes {M}_n$. Let $d\pi$ and $d\pi_0$,  be the representations of  Lie algebras  associated with  $\pi$  and $\pi_0$ (respectively) on the space of smooth vectors of $\mathcal{H}$.
Then the constant family of representations  $\{d\pi_{\epsilon}=d\pi \}_{\epsilon\in \R^*}$ 
  strongly contract to  $d\pi_{0}$.
}\\

 The second purpose  is to obtain the full unitary dual of   $\mathbb{R}^{2}\rtimes M_2$ as a strong contraction. More precisely we prove the following result.\\

\noindent \textbf{Theorem 2.}
\textit{For any unitary irreducible representation $\pi_0$  of $\mathbb{R}^{2}\rtimes M_2$ the following hold.  
\begin{enumerate}
\item There is  a realization of $\pi_0$ on  a Hilbert space of functions  $\mathcal{H}$.
\item There is a dense $(\mathbb{R}^{2}\rtimes \mathfrak{m}_2)$-invariant subspace $\mathcal{H}^{\infty}$ of $\mathcal{H}$. 
\item There is a family  of representations $\{d\pi_{\epsilon}:\g_2\longrightarrow \operatorname{End}(\mathcal{H}^{\infty}) \}_{\epsilon\in \R\neq 0}$,   or a sequence of representations $\{d\pi_{\epsilon_n}:\mathfrak{gl}_2(\R)\longrightarrow \operatorname{End}(\mathcal{H}^{\infty}) \}_{n\in \mathbb{N}}$
that  strongly contracts to  $d\pi_{0}:\mathbb{R}^{2}\rtimes \mathfrak{m}_2\longrightarrow \operatorname{End}(\mathcal{H}^{\infty})$.
\end{enumerate}}

\section{Contractions}
In this section we  introduce  notations and review some generalities on {I}n\"{o}n\"{u}-Wigner contraction \cite{Inonu-Wigner53}. We  spell out an {I}n\"{o}n\"{u}-Wigner contraction \cite{Inonu-Wigner53} of $\mathfrak{gl}_n(\mathbb{R})$ with respect to the Lie algebra of the mirabolic group $M_n$. Then we recall the notion  of  strong contraction of  representations of Lie algebras. 
\subsection{Contraction of Lie algebras}\label{sec}
Given a real Lie algebra $\g=(V,[\_,\_])$ (with $V$ the underlying vector space ,and $[\_,\_]$ the Lie brackets of $\g$) and a decomposition  $V=\mathfrak{k}\oplus \mathfrak{s}$ with $\mathfrak{k}$  being a subalgebra and $\mathfrak{p}$ a vector space complement to $\mathfrak{k}$, there is a corresponding   {I}n\"{o}n\"{u}-Wigner contraction; For every nonzero $\epsilon \in \R$ we have a linear invertible map $t_{\epsilon}:V\longrightarrow V$ given by $t_{\epsilon}(X_{\mathfrak{k}}+X_{\mathfrak{s}})=X_{\mathfrak{k}}+\epsilon X_{\mathfrak{s}}$,  for  $X_{\mathfrak{k}}\in \mathfrak{k}$ and  $X_{\mathfrak{s}}\in \mathfrak{s}$. For  $X,Y\in V$, the formula  $[X,Y]_{\epsilon}:=t_{\epsilon}^{-1}[t_{\epsilon}(X),t_{\epsilon}(Y)]$ defines  Lie brackets on $V$. We denote the  corresponding Lie algebra by  $\g_{\epsilon}=(V,[\_,\_]_{\epsilon})$.  
Moreover, for every $X,Y\in V$,
\[ [X,Y]_{0}:=\lim_{\epsilon \longrightarrow 0}t_{\epsilon}^{-1}[t_{\epsilon}(X),t_{\epsilon}(Y)] \] 
converges and defines Lie brackets on $V$. The obtained Lie algebra $\g_{0}=(V,[\_,\_]_{0})$ is called the contraction of $\g=\g_{1}$. The Lie algebra $\g_0$ is a semidirect product of $\mathfrak{k}$ and the  abelian ideal $\mathfrak{s}$.   This contraction is  denoted by $\g \stackrel{t_{\epsilon}}{\rightarrow}\g_0$. 

\subsection{The case of $\mathfrak{m}_n\subset \mathfrak{gl}_n(\mathbb{R})$ }
Consider the Lie algebra of $n\times n$ real matrices, $\mathfrak{gl}_n(\R)$. We denote its underlying vector space by $V_n$, and we shall use the standard basis $\{e_{i,j} | 1 \leq i,j \leq n\}$ of $V_n$ to define a subalgebra  and a vector space complement. As the subalgebra $\mathfrak{k}$  we take $\mathfrak{m}_n=\operatorname{span}_{\mathbb{R}}\{e_{i,j}| 1\leq i\leq n-1, 1\leq j \leq n \}$ and choose the complement $
\mathfrak{s}_n:=\operatorname{span}_{\mathbb{R}}\{e_{nj}| 1\leq j \leq n \}$.
For $\epsilon \neq 0$, the corresponding contraction maps $t_{\epsilon}:V_n\longrightarrow V_n$ are given by 
\begin{eqnarray}\nonumber
&& t_{\epsilon}(e_{ij})=\begin{cases}
e_{ij} & , i\neq n\\
\epsilon e_{ij}  &, i=n.
  \end{cases}
\end{eqnarray} 
For every $\epsilon \in \R$, the Lie brackets $[\_,\_]_{\epsilon}$ on $V_n$ are  explicitly given by 
\begin{eqnarray}\nonumber
&&[e_{ij},e_{kl}]_{\epsilon}=\begin{cases}
[e_{ij},e_{kl}]=\delta_{jk}e_{il}-\delta_{li}e_{kj} &, i\neq n, j\neq n \\
[e_{ij},e_{kl}]=\delta_{jk}e_{il}-\epsilon \delta_{li}e_{kj} &, i=n, j\neq n\\
[e_{ij},e_{kl}]=\epsilon\delta_{jk}e_{il}- \delta_{li}e_{kj} &, i\neq n, j= n\\
\epsilon[e_{ij},e_{kl}]=\epsilon(\delta_{jk}e_{il}-\delta_{li}e_{kj}) &, i=k=n.
\end{cases} 
\end{eqnarray}
The group $\mathbb{R}^{n}\rtimes M_n$ with product given by 
\[(v,A)(u,B)=(v+(A^{-1})^Tu,AB) \] 
for $(v,A),(u,B) \in \mathbb{R}^{n}\times M_n$,  is a Lie group containing $M_n$ as a subgroup and having $(\mathfrak{gl}_n(\R))_0$ as its Lie algebra.

\subsection{Strong contractions}
Below we recall the definition of \textit{strong contraction} in the special case of a common underlying inner product space of functions, for all representations that take part in the contraction procedure. This will be enough for the purposes of this paper. For the general definition see \cite{MR3687831}. 

Keeping the previous notation, suppose that $t_{\epsilon}:V\longrightarrow V$ is a family  of linear invertible  maps  that realizes a contraction from $\g=\mathfrak{g}_{1}=(V,[\_\_,\_\_]_1)$ to  $\mathfrak{g}_{0}=(V,[\_\_,\_\_]_0)$. Let $X$ be a topological space, and $\mu$ a posItive Borel measure on $X$.  Let $W$ be a subspace  of $L^2(X,d\mu)$.  Let  $\pi_{0}:\g_0\longrightarrow \operatorname{End}(W)$ be a representation of  $\g_0$ and   for every $\epsilon \neq 0 $,  $\pi_{\epsilon}:\g\longrightarrow \operatorname{End}(W)$ a representation  of $\g$.
The representation $\pi_0$ is a \textit{strong contraction} of the family $\{\pi_{\epsilon} \}_{\epsilon \neq 0}$ if the following hold.
\begin{enumerate}
\item For every $f\in W$, every $Y \in V$  and every $x\in X$, 
\[\lim_{\epsilon\rightarrow 0}\left( \pi_{\epsilon}(t_{\epsilon}Y)f\right) (x)=\left( \pi_{0}(Y)f\right) (x).\]
\item  For every $f\in W$ and every $Y \in V$,   
\[\lim_{\epsilon\rightarrow 0}\|\left( \pi_{\epsilon}(t_{\epsilon}Y)f\right)- \left( \pi_{0}(Y)f\right) \|=0.\]
\end{enumerate}  
Similarly, a  sequence of representations  $\pi_{n}:\g\longrightarrow \operatorname{End}(W)$ strongly contract to  $\pi_{0}:\g_0\longrightarrow \operatorname{End}(W)$ if theres is a sequence of real numbers $\epsilon_n$ converging to zero such that 
\begin{enumerate}
\item For every $f\in W$, every $Y \in V$  and every $x\in X$, 
\[\lim_{n \rightarrow \infty}\left( \pi_{n}(t_{\epsilon_n}Y)f\right) (x)=\left( \pi_{0}(Y)f\right) (x).\]
\item  For every $f\in W$ and every $Y \in V$,   
\[\lim_{n\rightarrow \infty}\|\left( \pi_{n}(t_{\epsilon_n}Y)f\right)- \left( \pi_{0}(Y)f\right) \|=0.\]
\end{enumerate}

\section{Strong contraction of constant families of unitary irreducible representations of $GL_n(\R)$ }
This section deals with strong contractions of constant families of unitary irreducible representations of $GL_n(\R)$. We shall start our discussion proving a simple lemma which is applicable in  more general context.

Let  $\g \stackrel{t_{\epsilon}}{\rightarrow}\g_0$ be a contraction of Lie algebras  with respect to a decomposition  $V=\mathfrak{k}\oplus \mathfrak{s}$ as above (recall that $V$ is the underlying vector space of $\g$). Let $G$ and $G_0$ be Lie groups with Lie algebras $\g$ and $\g_0$ respectively, and $K$ a common subgroup of $G$ and $G_0$ with Lie algebra $\mathfrak{k}$. Further assume that $G_0=K\ltimes S$ with $S$ an abelian vector group having  $\mathfrak{s}$ as its Lie algebra.  
With any  unitary  representation $\pi$ of $G$ on a Hilbert space $\mathcal{H}$ we can associate a unitary representation $\pi|_{K}$ of $K$ on  $\mathcal{H}$, simply by restriction.   With $\pi|_{K}$ we associate a unique  unitary  representation $\pi_0$ of  $G_0$ on $\mathcal{H}$ such that    $\pi_0|_{K}=\pi|_{K}$ and $\pi_0(s)$ is the identity operator for any $s\in S$.    In this section  we let $\mathcal{H}^{\infty}$ be the space of  $G$-smooth vectors in $\mathcal{H}$. It is known to be dense (e.g., see \cite[Thm. 3.15]{Knapp}) and via differentiation it carries representations $d\pi$ and $d\pi_0$ of  $\g$ and  $\g_0$, respectively.

\begin{lemma}\label{lem1}
Let  $\pi:G\longrightarrow \mathcal{U}(\mathcal{H})$ be a unitary irreducible representation, realized on a Hilbert space of functions $\mathcal{H}$.
Then  the constant family of representations  $\{d\pi_{\epsilon}:\g\longrightarrow \operatorname{End}(\mathcal{H}^{\infty}) \}_{\epsilon\in \R\neq 0}$ 
with $d\pi_{\epsilon}=d\pi$,  strongly contracts to  $d\pi_{0}:\g_0\longrightarrow \operatorname{End}(\mathcal{H}^{\infty})$.
\end{lemma}
The proof  is elementary and is  included just for completeness.
\begin{proof}
We can assume that $\mathcal{H}$ is a subspace of $L^2(X)$ for some measurable space $X$ with a positive Borel measure $\mu$. For every $Y\in \mathfrak{k}$, every  function $f\in \mathcal{H}$  and every $x\in X$, we have
\begin{eqnarray}\nonumber
&&\lim_{\epsilon\rightarrow 0}\left( d\pi_{\epsilon}(t_{\epsilon}Y)f\right) (x)=\left( d\pi(Y)f\right) (x)=\left( d\pi_0(Y)f\right) (x),\\ \nonumber
&&\lim_{\epsilon\rightarrow 0}\|\left( d\pi_{\epsilon}(t_{\epsilon}Y)f\right)- \left( d\pi_{0}(Y)f\right) \|=0.
\end{eqnarray} 
For every $Y\in \mathfrak{s}$, every  function $f\in \mathcal{H}$  and every $x\in X$, we have
\begin{eqnarray}\nonumber
&&\lim_{\epsilon\rightarrow 0}\left( d\pi_{\epsilon}(t_{\epsilon}Y)f\right) (x)=\lim_{\epsilon\rightarrow 0}\epsilon\left( d\pi(Y)f\right) (x)=0=\left( d\pi_0(Y)f\right) (x),\\ \nonumber
&&\lim_{\epsilon\rightarrow 0}\|\left( d\pi_{\epsilon}(t_{\epsilon}Y)f\right)- \left( d\pi_{0}(Y)f\right) \|=\lim_{\epsilon\rightarrow 0}\|\left( \epsilon d\pi(Y)f\right)-0 \|=\lim_{\epsilon\rightarrow 0}|\epsilon| \|\left(  d\pi(Y)f\right) \|=0.
\end{eqnarray} 
\end{proof}

We keep the above notation and let $G=GL_n(R)$, $K=M_n$ and  $S=\R^n$. Then the Kirillov conjecture implies that $\pi|_{M_n}$ (and hence also  $\pi_0$) is unitary irreducible for any unitary irreducible $\pi$. In general, that is for other groups,  $\pi_0$ is typically reducible.  Lemma \ref{lem1} implies the following result for the above mentioned contraction with respect to the mirabolic subgroup.

\begin{theorem}\label{th1}
Let  $\pi:GL_n(\R)\longrightarrow \mathcal{U}(\mathcal{H})$ be a unitary irreducible representation, realized on a Hilbert space of functions $\mathcal{H}$.
Then $\pi_0$ is unitary irreducible and  the constant family of representations  $\{d\pi_{\epsilon}:\mathfrak{gl}_n(\R)\longrightarrow \operatorname{End}(\mathcal{H}^{\infty}) \}_{\epsilon\in \R\neq 0}$ 
with $d\pi_{\epsilon}=d\pi$,  strongly contracts to  $d\pi_{0}:\mathbb{R}^{n}\rtimes \mathfrak{m}_n\longrightarrow \operatorname{End}(\mathcal{H}^{\infty})$.
\end{theorem}

\section{The unitary dual of  $\mathbb{R}^{2}\rtimes M_2$}\label{s4}
In this section, using the Mackey machine \cite{MR0031489,MR0044536,MR0396826}, we describe the unitary dual of $\mathbb{R}^{2}\rtimes M_2$. We shall explicitly write the associated representation of $\g_0=\mathbb{R}^{2}\rtimes \mathfrak{m}_2$ on a corresponding dense subspace of the space of  smooth vectors. These realizations are  used in section \ref{Chap4}.

The character group $\widehat{\mathbb{R}^{2}}$ consists of all  
functions of the form  \begin{eqnarray} \nonumber
&&\chi_{u}:\R^2\longrightarrow \C \\ \nonumber
&& \chi_{u}\left(v\right)=e^{i\langle u,v\rangle} ,
\end{eqnarray}
with $u,v\in \R^2$. The mirabolic group  $M_2$ acts on $\widehat{\mathbb{R}^{2}}$ by $A\cdot \chi_{u}=\chi_{Au}$. Unitary irreducible representations of $\mathbb{R}^{2}\rtimes M_2$ are parameterized by an orbit of some $\chi_{u}$ in  $\widehat{\mathbb{R}^{2}}$  and a unitary irreducible representation of the stabilizer of $\chi_{u}$. Below we give an exhaustive list of these representations up to equivalence. \begin{enumerate}
\item \textbf{The orbit of the  trivial character $\chi_{(0,0)}$}. In this case we obtain unitary irreducible representations that are trivial on $\R^2$. Such a representation is given by  a unitary irreducible representation of $M_2$. Explicitly,  we have exactly the following cases.
\begin{enumerate}
	\item For  $\lambda \in \mathbb{R}$ and $\sigma\in \{0,1\}$, the representation $\eta_{(0,0)}^{\lambda,\sigma}:	M_2\longrightarrow GL(\mathbb{C})\cong \mathbb{C}^*$   given by \[\eta_{(0,0)}^{\lambda,\sigma}\left(\begin{matrix}
a& b   \\
0 &  1   \end{matrix}\right) =\operatorname{sgn}(a)^{\sigma}|a|^{i\lambda}\mathbb{I}.\]
The associated representation $d\eta_{(0,0)}^{\lambda,\sigma}$ of $\g_0$ on $\C$ is given by 
\begin{eqnarray}\nonumber
&& e_{11}\longmapsto i\lambda,\hspace{2mm} e_{12} \longmapsto 0,\hspace{2mm} e_{21}\longmapsto 0, \hspace{2mm} e_{22} \longmapsto 0. 
\end{eqnarray}  
	\item The   representation $\eta_{(0,0)}:M_2 \longrightarrow \mathcal{U}(L^2(\mathbb{R^*},\frac{dx}{|x|}))$  given by 
	\[\left(\eta_{(0,0)}\left(\begin{matrix}
a& b   \\
0 &  1   \end{matrix}\right) f\right)(x)=e^{ibx} f(a x).\]
The associated representation $d\eta_{(0,0)}$ of $\g_0$ on $ C^{\infty}_c(\mathbb{R^*},\frac{dx}{|x|})$, the inner product  space of smooth compactly supported functions in $L^2(\mathbb{R^*},\frac{dx}{|x|})$, is given by 
\begin{eqnarray}\nonumber
&& e_{11}\longmapsto x\partial_x,\hspace{2mm} e_{12} \longmapsto i  x,\hspace{2mm} e_{21}\longmapsto 0, \hspace{2mm} e_{22} \longmapsto 0. 
\end{eqnarray}  
\end{enumerate}
\item  \textbf{The orbit of the   character $ \chi_{(1,0)} $.}  For $\lambda\in \R$, the  representation $\eta_{(1,0)}^{\lambda}$  on $L^2(\R^*,\frac{dx}{|x|})$  given by 
\[\left({\eta}^{\lambda}_{(1,0)}  \left( \left(\begin{matrix}
v_1   \\
v_2    \end{matrix}\right),\left(\begin{matrix}
a& b   \\
0 &  1    \end{matrix}\right) \right)f \right)(x)=e^{iv_1x^{-1}}e^{i  \lambda b x} f( x a ).\]
The associated representation $d\pi_{(1,0)}^{\lambda}$ of $\g_0$ on $ C^{\infty}_c(\mathbb{R^*},\frac{dx}{|x|})$, is given by 
\begin{eqnarray}\nonumber
&& e_{11}\longmapsto x\partial_x,\hspace{2mm} e_{12} \longmapsto i \lambda x,\hspace{2mm} e_{21}\longmapsto \frac{i}{x}, \hspace{2mm} e_{22} \longmapsto 0. 
\end{eqnarray}  

\item  \textbf{The orbit of the   character $ \chi_{(0,\beta)}$, with $\beta\neq 0$.} For $\lambda\in \R$ and $\sigma\in \{0,1\}$, the representation $\eta^{\lambda,\sigma}_{(0,\beta)}$   on 
$L^2\left(\R,dx\right)$  given by 
\[\left(\eta^{\lambda,\sigma}_{(0,\beta)} \left( \left(\begin{matrix}
v_1   \\
v_2   \end{matrix}\right),\left(\begin{matrix}
a& b   \\
0 &  1    \end{matrix}\right) f\right)\right)(x)=e^{i\beta(-xv_1+v_2)}\operatorname{sgn}(a)^{\sigma}|a|^{i\lambda-1/2} f\left( \frac{b+x}{a}\right).\]
The associated representation $d{\eta}_{(0,\beta)}^{\lambda, \sigma}$ of $\g_0$ on $ C^{\infty}_c(\mathbb{R},dx)$,  the inner product  space of smooth compactly supported functions in $L^2(\mathbb{R},dx)$,    is given by 
\begin{eqnarray}\nonumber
&& e_{11}\longmapsto   i\lambda -\frac{1}{2} - x\partial_x,\hspace{2mm} e_{12} \longmapsto \partial_x,\hspace{2mm} e_{21}\longmapsto -i \beta x, \hspace{2mm} e_{22} \longmapsto i\beta. 
\end{eqnarray}  
\end{enumerate}

We summarize the above discussion in a Lemma. 

\begin{lemma}\label{lem2}
The list below contains exactly one representative from  each equivalence class of  unitary irreducible representations of $\mathbb{R}^{2}\rtimes M_2$:
\begin{enumerate}
\item $\{\eta_{(0,0)}^{\lambda,\sigma}| \lambda \in \mathbb{R}, \sigma\in \{0,1\}\}\cup \{ \eta_{(0,0)}\}$,
\item $\{\eta_{(1,0)}^{\lambda}| \lambda \in \mathbb{R}\}$,
\item $\{\eta^{\lambda,\sigma}_{(0,\beta)} | \beta\in \R^*, \lambda \in \mathbb{R}, \sigma\in \{0,1\}\}$.
\end{enumerate}
\end{lemma}

\section{The unitary dual of  $\mathbb{R}^{2}\rtimes M_2$   as strong contraction}\label{Chap4}
In this section for any unitary irreducible representation of $\mathbb{R}^{2}\rtimes M_2$, in one of  the explicit realizations that are  given in Section \ref{s4}, 
 we build a corresponding strong contraction and, by doing so, proving Theorem \ref{th2}. We shall organize our calculations according to the orbits of $M_2$ on $\widehat{\R^2}$.  
\subsection{The orbit of the  trivial character $\chi_{0}$ }
For every $\mu \in \R$, we let $\pi_{\mu}$ be the unitary one dimensional representation of $GL_2(\R)$ given by 
\[ A\longmapsto |\det A|^{i\mu}. \]
The corresponding representation $d\pi_{\mu}$ of $\mathfrak{gl}_2(\R)$ is given by
\begin{eqnarray}\nonumber
&& e_{11}\longmapsto i\mu,\hspace{2mm} e_{12} \longmapsto 0,\hspace{2mm} e_{21}\longmapsto 0, \hspace{2mm} e_{22} \longmapsto i\mu. 
\end{eqnarray}  
  
\begin{proposition}\label{p1}
For every $\lambda\in \R$ and $\sigma\in \{0,1\}$,  the representation  $d\eta_{(0,0)}^{\lambda,\sigma}:\mathbb{R}^{2}\rtimes \mathfrak{m}_2 \longrightarrow \operatorname{End}(\C)$ is a strong contraction of the (constant) family $\{ d\pi_{\lambda}:\mathfrak{gl}_2(\R) \longrightarrow \operatorname{End}(\C)\}_{\epsilon \neq 0 }$.
\end{proposition}  
  We first remark that in this case we can  think of the one dimensional vector space $\C$ as the space of square integrable functions on the trivial measurable space consisting of one point only.  
\begin{proof}
In this case, since the underlying vector space is one dimensional  pointwise convergence implies norm convergence. Indeed,
\begin{eqnarray}\nonumber
&& \lim_{\epsilon \longrightarrow 0 } d\pi_{\lambda}:(t_{\epsilon}(e_{11})) =  i\lambda= d\eta_{(0,0)}^{\lambda,\sigma}(e_{11}),  \\ \nonumber
&& \lim_{\epsilon \longrightarrow 0 } d\pi_{\lambda}:(t_{\epsilon}(e_{12})) =0= d\eta_{(0,0)}^{\lambda,\sigma}(e_{12}),  \\ \nonumber
&& \lim_{\epsilon \longrightarrow 0 } d\pi_{\lambda}:(t_{\epsilon}(e_{21})) =0= d\eta_{(0,0)}^{\lambda,\sigma}(e_{21}),  \\ \nonumber
&& \lim_{\epsilon \longrightarrow 0 } d\pi_{\lambda}:(t_{\epsilon}(e_{22})) = \lim_{\epsilon \longrightarrow 0 } \epsilon  i\lambda=0= d\eta_{(0,0)}^{\lambda,\sigma}(e_{22}).  
\end{eqnarray} 
\end{proof}

\begin{proposition}
There is a unitary irreducible representation $\pi$  of $GL_2(\R)$ on   a Hilbert space of functions  $\mathcal{H}$ and a dense $\mathfrak{gl}_2(\R)$-invariant subspace  $\mathcal{H}^{\infty}$ of  $\mathcal{H}$,  such that the (constant) family of representations  $\{ d\pi:\mathfrak{gl}_2(\R) \longrightarrow \operatorname{End}(\mathcal{H}^{\infty})\}_{\epsilon \neq 0 }$ strongly contract to a representation isomorphic to $d\eta_{(0,0)}:\mathbb{R}^{2}\rtimes \mathfrak{m}_2 \longrightarrow \operatorname{End}(C^{\infty}(\mathbb{R^*},\frac{dx}{|x|}))$.
\end{proposition}  
\begin{proof}
We can take for $\pi$ any unitary irreducible infinite-dimensional representation of $GL_2(\R)$ realized on a Hilbert space of functions. For example, we can take a unitary irreducible principal series with $\mathcal{H}$  a Hilbert space of functions on $GL_2(\R)$. Theorem \ref{th1} guarantees that $\pi_0$ is unitary irreducible and  the constant family of representations  $\{d\pi_{\epsilon}:\g_2\longrightarrow \operatorname{End}(\mathcal{H}^{\infty}) \}_{\epsilon\in \R\neq 0}$ 
with $d\pi_{\epsilon}=d\pi$,  strongly contracts to  $d\pi_{0}:\mathbb{R}^{2}\rtimes \mathfrak{m}_2\longrightarrow \operatorname{End}(\mathcal{H}^{\infty})$. The representation $\pi_0|_{M_2}$ of $M_2$ is unitary irreducible and infinite-dimensional. Up to equivalence there is exactly one such representation. Since $\R^2\rtimes 1$ acts trivially via   $\pi_0$ then   $\pi_0$ must be equivalent to $\eta_{(0,0)}$.  Hence $d\pi_0$ is equivalent to $d\eta_{(0,0)}$.
\end{proof}

\subsection{The orbit of the   character $ \chi_{(1,0)} $.}
For every integer $n>1$  there is a discrete series representation  $D^{n}$ of $GL_2(\C)$ in which the scalar matrices act trivially.  In the  Kirillov model on $L^2(\R^*,\frac{dx}{|x|})$ the corresponding representation of $\mathfrak{gl}_2(\R)$  on   $C^{\infty}_c(\mathbb{R^*},\frac{dx}{|x|})$ is given by 
\begin{eqnarray}\nonumber
&& e_{11}\longmapsto x\partial_x, \hspace{1mm}e_{12}\longmapsto i x, \hspace{1mm} e_{21} \longmapsto -i\frac{n^2-1}{4x} +ix\partial_{xx},\hspace{1mm} e_{22}\longmapsto -x\partial_x.
\end{eqnarray}
For more details, see \cite{MR3055523}.
For $ 0 \neq q\in \R$ we  twist the above mentioned representation via conjugation by the diagonal matrix $\operatorname{diag}(q,1)$  to obtain the isomorphic representation $d{D^{n,q}}$ given by 
\begin{eqnarray}\nonumber
&& e_{11}\longmapsto x\partial_x, \hspace{1mm}e_{12}\longmapsto iq x, \hspace{1mm} e_{21} \longmapsto -i\frac{n^2-1}{4q x} +i\frac{x}{q}\partial_{xx},\hspace{1mm} e_{22}\longmapsto  -x\partial_x.
\end{eqnarray}

\begin{proposition}\label{pro3}
For any  $0\neq \lambda\in \R$   the representation $d\eta_{(1,0)}^{\lambda}:\mathbb{R}^{2}\rtimes \mathfrak{m}_2\longrightarrow \operatorname{End}( C^{\infty}_c(\mathbb{R^*},\frac{dx}{|x|}))$ is a strong contraction of the sequence  $\{ d{D^{n,\lambda}}:\mathfrak{gl}_2(\R) \longrightarrow \operatorname{End}( C^{\infty}_c(\mathbb{R^*},\frac{dx}{|x|}))\}_{n\in \mathbb{N}  }$.
\end{proposition}    
\begin{proof} 
We shall use the sequence  $\epsilon_n=-\frac{4\lambda}{n^2 }$. For pointwise convergence observe that for every $f\in C^{\infty}_c(\mathbb{R^*})$ and $x\in \R$, 
\begin{eqnarray}\nonumber
&& \lim_{n\longrightarrow \infty } \left(d{D^{n,\lambda}}(t_{\epsilon_n}(e_{11}))f\right)(x) =  \lim_{n\longrightarrow \infty } \left(x\partial_x\right)f(x)= x\partial_xf(x)=\left(d\eta_{(1,0)}^{\lambda}(e_{11})f\right)(x),\\ \nonumber
&& \lim_{n\longrightarrow \infty } \left(d{D^{n,\lambda}}(t_{\epsilon_n}(e_{12}))f\right)(x) =  i \lambda xf(x)=\left(d\eta_{(1,0)}^{\lambda}(e_{12})f\right)(x),\\ \nonumber
&& \lim_{n\longrightarrow \infty } \left(d{D^{n,\lambda}}(t_{\epsilon_n}(e_{21}))f\right)(x) =\lim_{n\longrightarrow \infty }  -\frac{4\lambda}{n^2 \lambda} \left( -i\frac{n^2-1}{4x} +ix\partial_{xx}\right)f(x)= \\ \nonumber
&& ix^{-1}f(x)=\left(d\eta_{(1,0)}^{\lambda}(e_{21})f\right)(x),\\ \nonumber
&& \lim_{n\longrightarrow \infty } \left(d{D^{n,\lambda}}(t_{\epsilon_n}(e_{22}))f\right)(x) =   \lim_{n\longrightarrow \infty } -\frac{4\lambda}{n^2 }   \left(-x\partial_x\right)f(x)=0= \left(d\eta_{(1,0)}^{\lambda}(e_{22})f\right)(x).
\end{eqnarray}
Norm convergence  follows from Lebesgue dominated convergence theorem. 
\end{proof}

\begin{proposition}
The representation $d\eta_{(1,0)}^{0}:\mathbb{R}^{2}\rtimes \mathfrak{m}_2\longrightarrow \operatorname{End}( C^{\infty}_c(\mathbb{R^*},\frac{dx}{|x|}))$ is a strong contraction of the sequence  $\{ d{D^{n,1/n}}:\mathfrak{gl}_2(\R) \longrightarrow \operatorname{End}( C^{\infty}_c(\mathbb{R^*},\frac{dx}{|x|}))\}_{\epsilon \neq 0 }$.
\end{proposition}   
Using $\epsilon_n=-\frac{4}{n^3}$ the proof is similar to that of Proposition \ref{pro3}.

\subsection{The orbit of the   character $\chi_{(0,\beta)}$ with  $0\neq \beta \in \R $.} 

In the non-compact picture, the unitary  principal series of $SL_2(\R)$  are realized on $L^2(\R,dx)$  via 
\[\left(\mathcal{P}^{\sigma,\nu} \left(\begin{matrix}
a& b   \\
c &  d    \end{matrix}\right) f\right)(x)=\operatorname{sgn}(-b x+d)^{\sigma}|-b x+d|^{-1-i\nu}f\left( \frac{ax-c}{-bx+d}\right),\]
here $\sigma\in \{0,1\}$ and $\nu\in \R$. See e.g., \cite[p. 36]{Knapp}.  We twist this representation by precompose $\mathcal{P}^{\sigma,\nu}$ with inverse transpose. Explicitly,  the new action is  $\widetilde{\mathcal{P}}^{\sigma,\nu}$  given by 
\[\left(\widetilde{\mathcal{P}}^{\sigma,\nu} \left(\begin{matrix}
a& b   \\
c &  d    \end{matrix}\right) f\right)(x)=\operatorname{sgn}(c x+a)^{\sigma}|c x+a|^{-1-i\nu}f\left( \frac{dx+b}{cx+a}\right).\]
 For every $\mu\in \R$, we  extend $\widetilde{\mathcal{P}}^{\sigma,\nu}$ to a unitary irreducible representation $\widetilde{\mathcal{P}}^{\sigma,\nu,\mu}$ of  $GL_2(\R)$,  by
 \[\left(\widetilde{\mathcal{P}}^{\sigma,\nu,\mu} \left(\begin{matrix}
a& b   \\
c &  d    \end{matrix}\right) f\right)(x)=|ad-bc|^{1/2+i\mu}\operatorname{sgn}(c x+a)^{\sigma}|c x+a|^{-1-i\nu}f\left( \frac{dx+b}{cx+a}\right).\]

 The corresponding  representation  $d\widetilde{\mathcal{P}}^{\sigma,\nu,\mu} $ of $\mathfrak{gl}_2(\R)$ on $C^{\infty}_c(\R)$ is given by  
\begin{eqnarray}\nonumber
&& e_{11}\longmapsto  i\mu-\frac{1}{2}-i\nu -x\partial_x,\hspace{1mm} e_{12}\longmapsto \partial_x, \\ \nonumber
&& e_{21} \longmapsto  -(1+i\nu)x -x^2\partial_{x}, \hspace{1mm}
 e_{22}\longmapsto    \frac{1}{2}+i\mu +x\partial_x. 
\end{eqnarray}

 \begin{proposition}\label{p5}
For $\lambda\in \R$ and $\sigma\in \{0,1\}$, the representation $d\pi^{\lambda,\sigma}_u:\mathbb{R}^{2}\rtimes \mathfrak{m}_2\longrightarrow \operatorname{End}( C^{\infty}_c(\R))$
 is a strong contraction of the family $\{ d\widetilde{\mathcal{P}}^{\sigma,\beta/\epsilon,\lambda+\beta/\epsilon}:\mathfrak{gl}_2(\R) \longrightarrow \operatorname{End}( C^{\infty}_c(\R))\}_{\epsilon \neq 0 }$.
\end{proposition} 

\begin{proof} For pointwise convergence observe that for every $f\in C^{\infty}_c(\mathbb{R^*})$ and $x\in \R$, 
\begin{eqnarray}\nonumber
&& \lim_{\epsilon \longrightarrow 0 } \left(d\widetilde{\mathcal{P}}^{\sigma,\beta/\epsilon,\lambda+\beta/\epsilon}(t_{\epsilon}(e_{11}))f\right)(x) =   \lim_{\epsilon \longrightarrow 0 }\left(i(\lambda+\frac{\beta}{\epsilon})-\frac{1}{2}-i\frac{\beta}{\epsilon} -x\partial_x \right)f(x)=\\ \nonumber
&&  \hspace{1cm} \left(i\lambda -\frac{1}{2} - x\partial_x\right)f(x)= \left(d\pi^{\lambda,\sigma}_u(e_{11})f\right)(x),  \\ \nonumber
&& \lim_{\epsilon \longrightarrow 0 } \left(d\widetilde{\mathcal{P}}^{\sigma,\beta/\epsilon,\lambda+\beta/\epsilon}(t_{\epsilon}(e_{12})) f\right)(x)=\partial_xf(x)= \left(d\pi^{\lambda,\sigma}_u(e_{12})f\right)(x),  \\ \nonumber
&&  \lim_{\epsilon \longrightarrow 0 } \left(d\widetilde{\mathcal{P}}^{\sigma,\beta/\epsilon,\lambda+\beta/\epsilon}(t_{\epsilon}(e_{21}))f\right)(x) = \lim_{\epsilon \longrightarrow 0 }\epsilon\left( -(1+i\frac{\beta}{\epsilon})x -x^2\partial_{x}\right)f(x) =\\ \nonumber
&& \hspace{1cm}-i\beta xf(x)=   \left(d\pi^{\lambda,\sigma}_u(e_{21})f\right)(x),  \\ \nonumber
&& \lim_{\epsilon \longrightarrow 0 }  \left(d\widetilde{\mathcal{P}}^{\sigma,\beta/\epsilon,\lambda+\beta/\epsilon}(t_{\epsilon}(e_{22}))f\right)(x) = \lim_{\epsilon \longrightarrow 0 } \epsilon \left(  \frac{1}{2}+ i(\lambda+\frac{\beta}{\epsilon})+x\partial_x\right)f(x)
= \\ \nonumber
&&\hspace{1cm}i\beta f(x)= \left(d\pi^{\lambda,\sigma}_u(e_{22})f\right)(x).  
\end{eqnarray}
Norm convergence  follows from the following lemma which easily  proved using  Lebesgue dominated convergence theorem. 
\begin{lemma}
For every  $\epsilon \in \R$, let $D_{\epsilon}$ be a smooth  differential operator  on $C^{\infty}_c(\R)$. Explicitly, $D_{\epsilon}=\sum_{i=0}^nd_{i,\epsilon}(x)\partial_x^i$
for some smooth functions $d_{i,\epsilon}(x)\in C^{\infty}(\R)$.  If for every $i\in \{1,2,,,n\}$,  and every $x\in \R$,   $\epsilon \longmapsto d_{i,\epsilon}(x)$  is a continuous  function (of $\epsilon$) then  for every $f\in   C^{\infty}_c(\R)$,
\[ \lim_{\epsilon \longrightarrow 0}\int_{\R}|D_{\epsilon}(f)(x)|^2dx=\int_{\R}|D_{0}(f)(x)|^2dx.\]
\end{lemma}
\end{proof}

Combining Lemma \ref{lem2} and Propositions \ref{p1}-\ref{p5} we obtaine the  following result.
\begin{theorem}\label{th2}
For any unitary irreducible representation $\pi_0$  of $\mathbb{R}^{2}\rtimes M_2$ the following hold.  
\begin{enumerate}
\item There is  a realization of $\pi_0$ on  a Hilbert space of functions  $\mathcal{H}$.
\item There is a dense $(\mathbb{R}^{2}\rtimes \mathfrak{m}_2)$-invariant subspace $\mathcal{H}^{\infty}$ of $\mathcal{H}$. 
\item There is a family  (or a sequence) of representations $\{d\pi_{\epsilon}:\mathfrak{gl}_2(\R)\longrightarrow \operatorname{End}(\mathcal{H}^{\infty}) \}_{\epsilon\in \R\neq 0}$ 
that  strongly contract to  $d\pi_{0}:\mathbb{R}^{2}\rtimes \mathfrak{m}_2\longrightarrow \operatorname{End}(\mathcal{H}^{\infty})$.
\end{enumerate}
\end{theorem}


\section*{References}
\bibliographystyle{vancouver}
\bibliography{references}


\end{document}